\documentclass[12pt,a4paper]{article}
\usepackage[utf8]{inputenc}
\usepackage[english]{babel}
\usepackage{amsmath}
\usepackage{amsfonts}
\usepackage{amssymb}
\usepackage{graphicx}
\usepackage{amsthm}
\usepackage{hyperref}

\newtheorem{lemma}{Lemma}

\newtheorem{theorem}[lemma]{Theorem}

\newtheorem{claim}[lemma]{Claim}

\author{Hans L. Bodlaender%
\thanks{Department of Information and Computing Sciences, Utrecht University, P.O.Box 80.089, 3508 TB Utrecht,
the Netherlands. 
Email: \href{mailto:H.L.Bodlaender@uu.nl}{H.L.Bodlaender@uu.nl},
and Department of Mathematics and Computer Science, Eindhoven University of Technology, Eindhoven,
the Netherlands. The research of this author was partially supported by the {\em Networks} project, supported by the Netherlands Organization for Scientific Research N.W.O.} 
\and
 Tom C. van der Zanden\thanks{Department of Information and Computing Sciences, Utrecht University, P.O.Box 80.089, 3508 TB Utrecht,
the Netherlands. 
Email: \href{mailto:T.C.vanderZanden@uu.nl}{T.C.vanderZanden@uu.nl}}
}

\begin{document}

\title{On Exploring Temporal Graphs of Small Pathwidth}

\date{}

\maketitle

\begin{abstract}
We show that the {\sc Temporal Graph Exploration Problem} is NP-complete, even when the underlying graph has pathwidth 2 and
at each time step, the current graph is connected.
\end{abstract}

\section{Introduction}
Networks can change during time: roads can be blocked or built, friendships can wither or new friendships are formed, connections in
a computer network can go down or be made available, etc. Temporal graphs can serve as a model for such changing networks. 

In this note,
we study the complexity of a problem on temporal networks: the {\sc Temporal Graph Exploration} problem. Recently, Akrida et al.~\cite{AkridaMS18} showed that this problem is NP-complete, even when the underlying graph is a star. An important special case, studied
by Erlebach et al.~\cite{ErlebachHK15}, is when at each point in time, the current graph is connected. This case is trivial when the
underlying graph is a tree; we show that it is already NP-complete when the underlying graph has pathwidth 2.

\smallskip

A {\em temporal graph} $\cal G$ is given by a
series of graphs $G_1=(V,E_1)$, $G_2 = (V,E_2)$, \ldots, $G_L=(V,E_L)$, each with the same vertex set, but the set of 
edges can be different at different time steps. At time step $i$, only the edges in $E_i$ exist and can be used.
Each $i$, $1\leq i\leq L$ is called a {\em time step}, $G_i$ is the {\em current graph} at time $i$. 
The {\em underlying} graph is formed by taking the union of the graphs at the different time steps.
I.e., if we have $L$ time steps, and graphs $G_1=(V,E_1)$, $G_2 = (V,E_2)$, \ldots, $G_L = (V,E_L)$, the underlying
graph is $(V, E_1 \cup E_2 \cup \cdots \cup E_L)$, so an edge exists in the underlying graph if it exists in at least one time step.
Many graph properties
can be studied in the setting of temporal graphs; this note looks at the problem of {\em exploring} the graph.

In temporal graphs, we can define a {\em temporal walk}:
we have an explorer who at time step
$1$ is at a specified vertex $s$; at each time step $i$ she can move over an edge in $G_i$ or remain at her current location.
In the {\sc Temporal Graph Exploration} problem, we are given a temporal graph and a starting vertex $s$, and are asked
if there exists a temporal walk starting at $s$ that visits all vertices within a given time $L$. 
A variant is when we require that the walk ends at the starting vertex $s$; we denote this by {\sc RTB Temporal Graph Exploration}, with
RTB the acronym of {\em return to base}. (See \cite{AkridaMS18}.)

Michail and Spirakis \cite{MichailS16} introduced the {\sc Temporal Graph Exploration} problem. 
It is easy to see that
even if the graphs do not change over time, the exploration problem is NP-complete, as it contains {\sc Hamiltonian Path} as a
special case (set $L=n-1$.) Michail and Spirakis \cite{MichailS16} showed that the problem does not have a $c$-approximation, unless $P=NP$,
and obtained approximation algorithms for several special cases.

An important special case is when we require that at each time step, the current graph $G_i$ is connected. Now, if the time $L$
is sufficiently large compared to the number of vertices $n$, it is always possible to explore the graph. Specifically,
Erlebach et al.~\cite{ErlebachHK15} showed that in this case, the graph can be explored in $O(t^2 n \sqrt{n} \log(n))$ time steps,
where $t$ is the treewidth of the underlying graph. Similarly, if the underlying graph is 
a $2$ by $n$ grid then $O(n \log^3 n)$ time steps always suffice.

Recently, Akrida et al.~\cite{AkridaMS18} studied the {\sc Temporal Graph Exploration} problem when the underlying graph is
a star $K_{1,r}$. Even when each edge exists in at most six time steps, the problem is NP-complete. We use the following
of their results as starting point.

\begin{theorem}[Akrida et al.~\cite{AkridaMS18}]
{\sc RTB Temporal Graph Exploration} is NP-complete, when the underlying graph is a star, and each edge exists in at most six graphs $G_i$,
and the start and end vertex is the center of the star.
\end{theorem}

For more results, including special cases, approximation algorithms and inapproximability results, see \cite{AkridaMS18,ErlebachHK15,MichailS16},
and see \cite{Michail16} for a survey.

\smallskip

It is well known that problems that are intractable (e.g., NP-hard) on general graphs become easier (e.g., linear time solvable)
when restricted to graphs of bounded treewidth
(see e.g., \cite[Chapter 7]{Cyganbook}.) An example is {\sc Hamiltonian Path}, which can be solved in $O(2^{O(t)} n)$ time on
graphs of treewidth $t$ \cite{BodlaenderCKN,CyganKN18}.
Unfortunately, these positive results appear not to carry over to temporal graphs: we show that the {\sc Temporal Graph Exploration} problem is NP-hard,
even when the underlying graph has pathwidth 2 (and thus also treewidth 2), and at each point in time, the graph is a tree, and thus
connected.

Interestingly, there are other problems on temporal graphs that do become tractable when the treewidth is bounded. Specifically,
Fluschnik et al.~\cite{FluschnikMNZ18} showed that finding a small {\em temporal separator} becomes tractable when the underlying
graph has bounded treewidth; the problem is NP-hard in general \cite{KempeKK02}.

The pathwidth of graphs was defined by Robertson and Seymour~\cite{RobertsonS1}.  A {\em path decomposition} of a graph $G=(V,E)$ is a sequence of subsets (called {\em bags}) of $V$ $(X_1, \ldots, X_r)$,
such that $\bigcup_{1\leq i\leq r} X_r=V$, for all $\{v,w\}\in E$, there is an $i$ with $v,w\in X_i$, and if $1\leq i_1<i_2<i_3\leq r$,
then $X_{i_2} \subseteq X_{i_1}\cap X_{i_3}$. The {\em width} of a path decomposition $(X_1, \ldots, X_r)$ equals $\max_{1\leq i\leq r} |X_i|-1$;
the {\em pathwidth}  of a graph $G$ is the minimum width of a path decomposition of $G$.
The pathwidth of a graph is an upper bound for its treewidth. (See e.g.~\cite[Chapter 7]{Cyganbook}.)

$K_{1,r}$ is a star graph with $r+1$ vertices, i.e., we have one vertex of degree $r$ which is adjacent to the remaining $r$ vertices, which have degree 1.

\section{Hardness result}
We now give our main result. 

\begin{theorem}
The {\sc Temporal Graph Exploration Problem} is NP-complete, even if each graph $G_i =(V,E_i)$ is a tree, and
the underlying graph has pathwidth 2.
\label{theorem:main}
\end{theorem}

\begin{proof}
We use a reduction from {\sc RTB Temporal Graph Exploration} for star graphs. Suppose we have a temporal star ${\cal K}_{1,n-1}$, given by
a series of subgraphs of $K_{1,n-1}$, $G_1=(V,E_1)$, $\ldots$, $G_L=(V,E_L)$, and a start vertex $s$, which is the center of the star. We denote the vertices of $K_{1,n}$ by $v_0, \ldots, v_{n-1}$, with $s=v_0$.

We now build a new temporal graph, as follows.
Set $Q = L \cdot (n+3)$.

The vertex set of the new graph consists of $V$ and $Q+1$ new vertices. These will form a path. The new
vertices are denoted $p_0, \ldots, p_{Q}$ and called {\em path vertices}; the vertices in $V$ are called {\em star vertices}.

We now define a temporal graph ${\cal G}'$, given by a series of graphs $G'_i$, $1\leq i \leq L'$. $G'_i$ has the following edges:
\begin{itemize}
 \item For each $i$, the vertices $p_0, \ldots, p_{Q}$ form a path: we have edges $\{p_j,p_{j+1}\}$ for $1\leq j< Q$.
 \item If $i\leq L$, all edges in $G_i$ are also edges in $G'_i$. 
 \item If $i\leq L$, for each star vertex $v_j\in V$: if $v_j$ is the lowest numbered vertex in a connected component of
 $G_i$, we have an edge $\{v_j, p_{L \cdot (j+2)}\}$.
 \item If $i>L$, we have an edge from each star vertex $v_i\neq s$ to $s$, and an edge from $s$ to $p_0$.
\end{itemize}

It is not hard to see that each $G'_i$ is a tree. If $i\leq L$, then $G'_i$ is obtained by adding the path to $G_i$ and one edge from the path to each
connected component of $G_i$. If $i\geq L$, then $G'_i$ is obtained taking a path and $K_{1,n}$ and adding an edge between a path and star vertex.

\begin{figure}[htb]
\begin{center}
\includegraphics[scale=1]{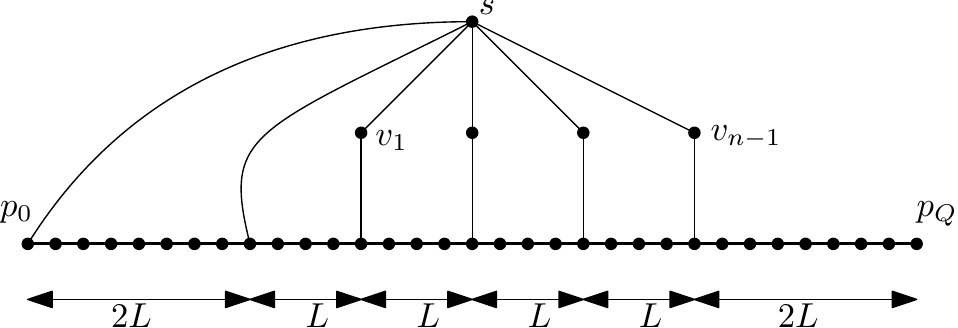}
\end{center}
\label{figure:example}
\caption{Illustration to the proof of Theorem~\ref{theorem:main}. Note that edges between path vertices are present at each time step; other edges are present in a subset of the time steps.}
\end{figure}

\smallskip

The idea behind the proof is that during the first $L$ time steps, we explore the star vertices as normal, while the path serves to keep the graph connected but can not be explored. To explore the path vertices, we must make one single pass from $p_0$ to $p_Q$, as we do not have sufficient time to traverse either the section from $p_0$ to $p_{2L-1}$ or that from $p_{Q-2L+1}$ twice: traversing the edges between star vertices and path vertices (other than edge $\{s,p_0\}$) cannot contribute to a solution.

\begin{lemma}
There is a temporal walk in ${\cal G}'$ that starts at $s$ and visits all vertices in $G'$ 
in at most $L+Q+1$ time steps, if and only if there is
a temporal walk in the temporal star ${\cal K}_{1,n-1}$ that starts at $s$, ends in $s$ and visits all vertices in ${\cal K}_{1,{n-1}}$ in at most $L$ time steps.
\label{lemmaiff}
\end{lemma}

\begin{proof}
First, suppose that there is
a temporal walk in  ${\cal K}_{1,n-1}$ that starts at $s$, ends at $s$ and visits all vertices in at most $L$ time steps. Then, we visit all vertices in ${\cal G}'$, by
first making the temporal walk in the star, if necessary wait in $s$ until the end of time step $L$ and at time $L+1$ move 
from $s$ to
$p_1$, and then visit all path vertices by traversing the path in the remaining $Q$ time steps.

\smallskip

Suppose we have a temporal walk that starts at $s$ and visits all vertices in ${\cal G}'$ in at most $L+Q+1$ time steps.

\begin{claim}
If we are at a path vertex $p_i$ at the end of time step $\alpha\leq L$, then $L < i < Q-L$.
\label{claim1}
\end{claim}

\begin{proof}
If we are at a path vertex $p_i$ at the end of time step $\alpha\leq L$, then we moved one or more times from a star vertex to a path vertex
during the first $\alpha$ time steps. Consider the last of these moves, say that we moved at time step $\beta \leq \alpha$ from a star vertex
$v_j$ to a path vertex $p_{j'}$; between time step $\beta+1$ and $\alpha$ we stay at path vertices. We have that $j'= L \cdot (j+2)$, by construction of the temporal graph. We can make less than $L$ steps after reaching $p_{j'}$ until time step $\alpha\leq L$, hence
$j'-L < i < j'+L$. Now, $L = L \cdot (0+2) -L \leq L \cdot (j+2) -L = j'-L < i < j'+L = L \cdot (j+2) +L \leq L \cdot (n-1+2) +L = (n+2)\cdot L = Q-L$.
\end{proof}

\begin{claim}
At the end of time step $L$, we must be in vertex $s$.
\label{claim2}
\end{claim}

\begin{proof}
Suppose not. Note that both $p_0$ and $p_Q$ are not yet visited, by claim~\ref{claim1}. If we are at a star vertex $v_i\neq s$ at the end of time step $L$, then
we can only visit the path vertices by first moving to $s$, then to $p_0$, and then visiting the path vertices in order; this costs one time step too many. Suppose we are at a path vertex $p_i$ at the end of time step $L$. Suppose we visit $p_0$ before $p_Q$. Then, 
 we must make
at least $i$ steps from $p_i$ to $p_0$, and then $Q$ steps from $p_0$ to $p_Q$: in total $i+Q > L+Q$ steps; contradiction. Suppose
we visit $p_Q$ before $p_0$. Then we must make at least $Q-i$ steps from $p_i$ to $p_Q$, and then $Q$ steps from $p_Q$ to $p_0$: in total
$2Q-i > L+Q$ steps; this is again a contradiction.
\end{proof}

\begin{claim}
If we move at time step $i \leq L$ from a star vertex $v_i$ to a path vertex $p_j$, then the first star vertex visited after time step $i$ is again $v_i$, and this move to $v_i$ will be made before the end of time step $L$.
\label{claim3}
\end{claim}

\begin{proof}
By Claim~\ref{claim2}, we must move to a star vertex before the end of time step $L$. If $p_{j'}$ is a neighbor of a star vertex and $j\neq j'$,
then $p_{j'}$ is at least $L$ steps on the path away from $p_j$, so we cannot reach $p_{j'}$ before time $L$, hence we must move back to the star from $p_j$, and thus move to $v_i$.
\end{proof}

Now, we can finish the proof of Lemma~\ref{lemmaiff}. Take from the walk in ${\cal G}'$ the first $L$ time steps. Change this by replacing each move
to a path vertex by a step where the explorer does not move. I.e., when the walk in ${\cal G}'$ moves from star vertex $v_i$ to a path vertex,
then we stay in $v_i$ until the time step where the walk in $G'$ moves back to the star --- by Claim~\ref{claim3}, this is a move to $v_i$.
In this way, we obtain a walk in ${\cal K}_{1,n-1}$ that visits all vertices in $L$ time steps.

\end{proof}

It remains to show that the underlying graph has pathwidth $2$. 
If we remove $s$ from the underlying graph, then we obtain a caterpillar: a graph that can obtained by taking a path, and adding vertices
of degree one, adjacent to a path vertex. These have pathwidth 1 \cite{ProskurowskiT99}; now add $s$ to all bags and we obtain a path decomposition of the underlying graph of ${\cal G}'$ of width 2.
\end{proof}

A minor variation of the proof gives also the following result.

\begin{theorem}
The {\sc RTB Temporal Graph Exploration Problem} is NP-complete, even if each graph $G_i =(V,E_i)$ is a tree, and
the underlying graph has pathwidth 2.
\label{theorem:main2}
\end{theorem}

\begin{proof}
Modify the proof of Theorem~\ref{theorem:main} as follows: add one time step; the current graph in the last time step has one edge,
from $p_{Q}$ to $s$.
\end{proof}

\section{Conclusions}
In this note, we showed that the {\sc Temporal Graph Exploration Problem} is NP-complete, even when we require that at each time step, the
graph is connected, or more specifically a tree, and the underlying graph (i.e., the graph where an edge exists whenever it exists for at least one time step) has pathwidth 2, and hence treewidth 2.
This contrasts many other results for graphs of bounded treewidth, including a polynomial time algorithm for finding small
temporal separators for graphs of small treewidth \cite{FluschnikMNZ18}. 

If we require that the graph is connected at each time step, the case that the treewidth is 1 becomes trivial (as this deletes all temporal
effects). Interesting open cases are when the underlying graph is outerplanar, or an almost tree, i.e, 
can be obtained by adding one edge to a tree.

%\bibliographystyle{abbrv}
%\bibliography{temporal}

\end{document}